\documentclass[aps,pra,reprint,twocolumn,superscriptaddress,showpacs,showkeys,a4paper]{revtex4-1}
\usepackage{amsfonts}
\usepackage{amssymb}
\usepackage{bbm, dsfont}
\usepackage[dvips]{graphicx}
\usepackage[T1]{fontenc}
\usepackage[latin2]{inputenc}
\usepackage{amsmath}
\usepackage{latexsym}
\usepackage{bbm}

\usepackage{color}

\def\dt{{\rm d}\,}

\def\duzomniejsze{<\kern-.7mm<}
\def\duzowieksze{>\kern-.7mm>}

\def\textbf#1{{\bf #1}}
\def\beq{\begin{equation}}
\def\eeq{\end{equation}}
\def\be{\begin{equation}}
\def\ee{\end{equation}}
\def\ben{\begin{eqnarray}}
\def\een{\end{eqnarray}}
 \def\beqa{\begin{eqnarray}}
\def\eeqa{\end{eqnarray}}
\def\eea{\end{array}}
\def\bea{\begin{array}}
\newcommand{\bei}{\begin{itemize}}
\newcommand{\eei}{\end{itemize}}
\newcommand{\bee}{\begin{enumerate}}
\newcommand{\eee}{\end{enumerate}}

\def\1{\openone}

\def\tr{{\rm Tr}}

\def\>{\rangle}
\def\<{\langle}

\def\dt#1{{{\kern -.0mm\rm d}}#1\,}
\def\squareforqed{\hbox{\rlap{$\sqcap$}$\sqcup$}}
\def\qed{\ifmmode\squareforqed\else{\unskip\nobreak\hfil
\penalty50\hskip1em\null\nobreak\hfil\squareforqed
\parfillskip=0pt\finalhyphendemerits=0\endgraf}\fi}

\newtheorem{lemma}{Lemma}
\newtheorem{theorem}[lemma]{Theorem}
\newtheorem{main result}[lemma]{Main result}
\newtheorem{proposition}[lemma]{Proposition}
\newtheorem{definition}[lemma]{Definition}

\newtheorem{fact}[lemma]{Fact}

\newtheorem{corollary}[lemma]{Corollary}

\newenvironment{proof}[1][Proof]{\textbf{#1.} }{\ \rule{0.5em}{0.5em}}

\def\bep{\begin{proposition}}
\def\eep{\end{proposition}}
\def\bel{\begin{lemma}}
\def\eel{\end{lemma}}

\def\bet{\begin{theorem}}
\def\eet{\end{theorem}}
\def\bed{\begin{definition}}
\def\eed{\end{definition}}
\def\bef{\begin{fact}}
\def\eef{\end{fact}}

\begin{document}

\title{Quantum-correlation breaking channels, broadcasting scenarios, and finite Markov chains}

\author{J.~K.~Korbicz}
   \email{Jaroslaw.Korbicz@icfo.es}
   \affiliation{ICFO-Institut de Ci\`encies Fot\`oniques, Av. Carl Friedrich Gauss 3, 08860 Castelldefels, Barcelona, Spain}
\author{P.~Horodecki}
   \affiliation{Faculty of Applied Physics and Mathematics, Gda\'nsk University of Technology, 80-233 Gda\'nsk, Poland}
   \affiliation{National Quantum Information Centre in Gdan'sk, 81-824 Sopot, Poland}
\author{R.~Horodecki}
   \affiliation{Institute of Theoretical Physics and Astrophysics, University of Gda\'nsk, 80-952 Gda\'nsk, Poland}
   \affiliation{National Quantum Information Centre in Gdan'sk, 81-824 Sopot, Poland}

\date{\today}

\begin{abstract}
One of the classical results concerning quantum channels is the characterization 
of entanglement-breaking channels [M. Horodecki \emph{et al.}, Rev. Math. Phys~{\bf 15}, 629 (2003)]. 
We address the question whether there exists a similar characterization
on the level of quantum correlations which may go beyond entanglement. 
The answer is fully affirmative in the case
of breaking quantum correlations down to the, so called, QC (Quantum-Classical) type, 
while it is no longer true in the CC (Classical-Classical) case. The corresponding channels turn out to be
measurement maps.
Our study also reveals an unexpected link between quantum state and local correlation broadcasting
and finite Markov chains. We present a possibility of broadcasting via non von Neumann
measurements, which relies on the Perron-Frobenius Theorem. Surprisingly, this is not the typical generalized C-NOT gate 
scenario, appearing naturally in this context.
\end{abstract}

\pacs{03.67.Hk, 03.67.Mn, 03.65.Ta, 02.50.Ga}

\keywords{quantum channels, quantum state and correlations broadcasting, finite Markov chains}

\maketitle
There is a well-known result concerning a characterization of entanglement-breaking channels
\cite{Michal,Holevo}. The latter are defined as channels which turn any bipartite state
(when applied to one subsystem) into a separable (non-entangled) one. 
The main result of Ref.~\cite{Michal} states that a channel $\Lambda$ is entanglement breaking
if and only if its Choi-Jamio\l kowski state (i.e. its witness) $\mathds{1}\otimes\Lambda (P_+)$ is a 
separable state ($P_+$ denotes the projector on the maximally entangled state, see Eq.~(\ref{P+})).   
However, it is
known that quantum correlations are more general
than entanglement (see e.g. Ref.~\cite{Modi} and references therein).

To our knowledge, the characterization from Ref.~\cite{Michal}
has not yet been refined to a case when a channel breaks more general quantum correlations, 
i.e. transforms any state into a state that does not possess some type of quantum correlations
(see however Ref.~\cite{Darek} where partial results were obtained).
Here we show that such a refinement is indeed possible for channels 
mapping (when applied to one subsystem) any bipartite state into a, so called, QC state.
Such channels turn out to be quantum-to-classical measurement maps \cite{broadcasting}.
Moreover, we show that a similar statement does not hold in the case of a stronger 
requirement of fully breaking quantum correlations and 
transforming any bipartite state into a CC form. In the latter case, which
is even more intriguing than the QC one, the corresponding measurement maps are formed by commuting 
Positive Operator Valued Measures (POVMs). 

Our study of QC-type channels leads to an unintuitive and surprising
connection between broadcasting of quantum states \cite{state-broadcasting} and correlations \cite{broadcasting,Luo} 
on one side, and finite Markov chains (see e.g. Ref.~\cite{Markov}) on the other. The existence of a broadcastable state
for a given QC-type channel is guaranteed by the fact that each finite Markov chain, described by a 
stochastic transition matrix \cite{HornJonson}, possesses by the Perron-Frobenius Theorem 
a stationary distribution.
In fact, it happens that there are
maps that may broadcast full rank states and still have the broadcasting 
restricted only to a convex subset of a full commuting family. Similar conclusion works for the 
case of broadcasting of correlations.
 
Recall that a QC (or more precisely Q$_A$C$_B$) state is a bipartite state of a form 
\be\label{def:QC}
\sigma^{QC}=\sum_i p_i \sigma_i^A \otimes |e_i\rangle_B\langle e_i|,
\ee
where $\sigma_i$'s are states at Alice's side, $\{e_i \}$ is an orthonormal basis on Bob's side (possibly different
from the computational basis $\{|i\rangle\}$), and $p_i$'s are probabilities. In the analogous way one defines
a CQ (ore more precisely C$_A$Q$_B$) state, where the classical part (projectors on the orthonormal basis) is located
at Alice's side.

Throughout the work we will always assume that $\Lambda$ is a trace-preserving, completely positive map, i.e. a channel, 
and 
\be\label{P+}
P_+:=|\psi_+\rangle\langle \psi_+|=\frac{1}{d}\sum_{i,j} |ii\rangle\langle jj|
\ee 
is the projector on the maximally entangled state $\psi_+$ and $\{|ij\rangle\}$ is a fixed computational product basis. 
We prove the following
\bet\label{thm:main}
For any channel $\Lambda$ its Choi-Jamio\l kowski state $\mathds{1}\otimes\Lambda (P_+)$ is a QC state if and only if  
$\mathds{1}\otimes\Lambda (\varrho_{AB})$ is a QC state for any bipartite state $\varrho_{AB}$.
\eet
\begin{proof}
We propose to call the above type of channels \emph{QC-type channels}.
In order to setup the notation and methods (cf. Ref.~\cite{Michal}), we present a detailed proof.
In one direction the implication is obvious. To prove it in the other one, 
assume that the state $\mathds{1}\otimes\Lambda (P_+)$ is QC
\be
\mathds{1}\otimes\Lambda (P_+)=\sum_i p_i \sigma_i \otimes |e_i\rangle\langle e_i|.\label{witness}
\ee
From the inversion formula for the Choi-Jamio\l kowski isomorphism \cite{CJ}
\be\label{CJinv}
\Lambda(A)=d\tr_A\big[W_\Lambda(A^T\otimes\mathds{1})\big],
\ee
where $W_\Lambda=\mathds{1}\otimes\Lambda(P_+)$ and the transposition is defined in the computational basis $\{|i\rangle\}$,
it follows that (\ref{witness}) is equivalent to
\be
\Lambda(\varrho)=d\sum_ip_i\tr(\varrho\sigma_i^T )|e_i\rangle\langle e_i|\label{Omega},
\ee
and hence
\be\label{1Lrho}
\mathds{1}\otimes\Lambda(\varrho_{AB})=d\sum_k p_k \tr_B(\varrho_{AB}\mathds{1}\otimes\sigma_k^T)\otimes|e_k\rangle\langle e_k|
\ee
for an arbitrary bipartite state $\varrho_{AB}$. We define unnormalized residual states
\be
\tilde\varrho_k^A:=d\,p_k\tr_B(\varrho_{AB}\mathds{1}\otimes\sigma_k^T),
\ee
and their traces
\be
\tilde p_k:=\tr\tilde\varrho_k^A=d\,p_k\tr_{AB}(\varrho_{AB}\mathds{1}\otimes\sigma_k^T).\label{pAk}
\ee
We show that $\sum_k \tilde p_k =1$. From the assumption  that $\Lambda$ is trace-preserving, it follows that
\beqa
&& \tr_B [\mathds{1}\otimes\Lambda (P_+)]=\frac{1}{d}\sum_{i,j}|i\rangle\langle j|\,\tr\Lambda(|i\rangle\langle j|)\nonumber\\
&& =\frac{1}{d}\sum_{i,j}|i\rangle\langle j|\,\tr|i\rangle\langle j|=\frac{1}{d}\sum_{i}|i\rangle\langle i|=\frac{\mathds{1}}{d}.
\label{TP}\eeqa
On the other hand, the QC assumption (\ref{witness}) implies that
\be
\tr_B [\mathds{1}\otimes\Lambda (P_+)]=\sum_kp_k\sigma_k,
\ee
and consequently
\be
\sum_kp_k\sigma_k=\frac{\mathds{1}}{d}.\label{sum_sigma}
\ee
Thus, the collection $\{dp_i\sigma_i\}$, or equivalently its transposition  
\be\label{Ei}
E_i:=dp_i\sigma_i^T, 
\ee
forms a POVM, which together with Eq.~(\ref{pAk}) implies that
\be
\sum_k \tilde p_k=d\tr_{AB}(\varrho_{AB}\mathds{1}\otimes\sum_kp_k\sigma_k^T)=\tr\varrho_{AB}=1.
\ee
Hence, Eq.~(\ref{1Lrho}) may be rewritten as
\be\label{finish}
\mathds{1}\otimes\Lambda(\varrho_{AB})=\sum_k\tilde p_k \varrho_k^A\otimes |e_k\rangle\langle e_k|,
\ee
with $\varrho_k^A:=\tilde \varrho_k^A/\tr\tilde \varrho_k^A=\varrho_k^A/\tilde p_k$, which is a QC state.
\end{proof}

We remark that Thm.~\ref{thm:main} will not in general be true if one changed the QC state to a CQ one, keeping
the form of the Choi-Jamio\l kowski isomorphism. Indeed, if $\mathds{1}\otimes\Lambda(P_+)
=\sum_i p_i |e_i\rangle\langle e_i|\otimes\sigma_i$, 
then from Eq.~(\ref{CJinv}) it follows that $\Lambda(\varrho)=d\sum_ip_i\langle e^*_i|\varrho|e^*_i\rangle\sigma_i$
and $\mathds{1}\otimes\Lambda(\varrho_{AB})=d\sum_i p_i \tr_B(\varrho_{AB}\mathds{1}\otimes|e^*_i\rangle\langle e^*_i|)\otimes\sigma_i$, 
which is in general a separable state but not a CQ nor QC one.
As an example, consider $\Lambda^{CQ}$ as a von Neumann measurement in the standard basis
on a qubit. Obviously, $\mathds{1} \otimes \Lambda(P_{+})$ is a CQ state, since it is CC.
Now consider a two-qubit state $\varrho_{AB}$ which is an unbiased mixture of the projectors
corresponding to two vectors $|\psi_{+}\rangle=1/\sqrt{2} (|00\rangle +|11\rangle ) $
and $|+ \rangle |0\rangle$ (here $|+ \rangle:=1/\sqrt{2}(|0\rangle +|1\rangle)\,$).
Then $\mathds{1} \otimes \Lambda(\varrho_{AB})=1/2\sum_{i=0,1} \varrho_{i} \otimes |i\rangle \langle i|$,
where $\varrho_{0}:=1/2 (|+\rangle \langle +|  + |0\rangle \langle 0|)$
and $\varrho_{1}:=|1\rangle \langle 1|$. But $[\varrho_{0},\varrho_{1}]\neq 0$,
breaking  the necessary condition for $\mathds{1} \otimes \Lambda(\varrho_{AB})$ to be a CQ state. 

As expected from the general results of Ref.~\cite{Michal} on entanglement breaking channels,
Eqs.~(\ref{Omega}), (\ref{sum_sigma}) and (\ref{Ei})
imply that the action of QC-type channel $\Lambda^{QC}$ consist of a POVM-measurement followed by a state preparation, 
but the preparation is always done in the same orthonormal basis $\{e_i\}$
\be\label{MP}
\Lambda(\varrho)=\sum_i\tr(\varrho E_i)|e_i\rangle\langle e_i|.
\ee
The later plays a role of
a classical register, so that every QC-type channel is in fact a quantum-to-classical measurement map \cite{broadcasting}:
$\Lambda(\varrho)$ gives the state of a measuring apparatus after the measurement of $\{E_i\}$ 
on a system in the state $\varrho$.  
In the light of this observation, Thm.~\ref{thm:main} states that a channel
is a measurement map if and only if (iff) its Choi-Jamio\l kowski state is a QC state.

A natural question arises if one can refine Thm.~\ref{thm:main} even more to the so-called CC states, i.e.
states of a form
\be\label{CC}
\sigma^{CC}=\sum_{i,j}p_{ij}^{AB} |e_i\rangle\langle e_i|\otimes |f_j\rangle\langle f_j|,
\ee
where now $\{e_i\}$ and $\{f_j\}$ are orthonormal bases on Alice's and Bob's side correspondingly, and $p_{ij}$ 
is a classical joint probability distribution. It turns out that as stated, Thm.~\ref{thm:main} does not
specify down to such a case, as even if $\mathds{1}\otimes\Lambda (P_+)$ is a CC state, $\mathds{1}\otimes\Lambda (\varrho_{AB})$
is generically a QC state. To see this, assume that
\be
\mathds{1}\otimes\Lambda (P_+)=\sum_{i,j}p_{ij} |e_i\rangle\langle e_i|\otimes |f_j\rangle\langle f_j|.\label{witness_cc_0}
\ee
From the inversion formula (\ref{CJinv}) one then obtains that
\beqa
&& \Lambda(\varrho)=\sum_j \tr(\varrho E_j)|f_j\rangle\langle f_j|,\label{LCC}\\
&& \mathds{1}\otimes\Lambda (\varrho_{AB})=\sum_j\tr_B(\varrho_{AB}\mathds{1}\otimes E_j)\otimes |f_j\rangle\langle f_j|,\label{witness_cc}
\eeqa
where now
\be\label{Ej}
E_j:=d\sum_i p_{ij} |e_i^*\rangle\langle e_i^*|,
\ee    
and the complex conjugation $e_i^*$ of the basis vectors $e_i$ is defined in the computational basis $\{|i\rangle\}$.

Similarly to the QC case, trace-preserving property of $\Lambda$ implies that $\{E_j\}$ form a POVM,
$\sum_jE_j=\mathds{1}$ (cf. Eqs.~(\ref{TP})-(\ref{sum_sigma})). 
However, in this case the POVM elements necessarily pairwise commute
\begin{equation}\label{Ejcommute}
[E_j,E_{j'}]=0,
\ee
since by Eq.~(\ref{Ej}) they correspond to a measurement in
one fixed basis, but they need not form a von Neumann measurement, as in general
$E_j$'s may overlap
\be
E_jE_{j'}=\sum_i p_{ij}p_{ij'}|e^*_i\rangle\langle e^*_i|\ne \delta_{jj'}E_j.
\ee
What is quite important is that the POVM condition $\sum_jE_j=\mathds{1}$, puts some constraints on $p_{ij}$:
\be
\sum_{i,j}p_{ij}|e_i^*\rangle\langle e_i^*|=\frac{\mathds{1}}{d} \Rightarrow p_i:=\sum_j p_{ij}=\frac{1}{d},
\ee
which in turn implies that the numbers
\be\label{p'}
p^\Lambda_{\phantom{i}j|i}:=dp_{ij}
\ee
are in fact conditional probabilities: $\sum_jp^\Lambda_{\phantom{i}j|i}=1$ for any $i$. 
Thus, the matrix $P^\Lambda:=[p^\Lambda_{\phantom{i}j|i}]$ is a stochastic matrix \cite{HornJonson} and 
\be\label{Ej'}
E_j=\sum_i p^\Lambda_{\phantom{i}j|i}|e_i^*\rangle\langle e_i^*|.
\ee

From a probabilistic point of view, a stochastic matrix defines a finite Markov chain \cite{Markov}: it provides
transition probabilities between the sites. Hence, with every CC-type channel satisfying (\ref{witness_cc_0}) there is
an associated finite Markov chain and vice versa---with every $d$-site Markov chain
and orthonormal bases $\{e_i\},\{f_i\}$
one can associate a CC-type channel through the formulas (\ref{LCC}) and (\ref{Ej'}). 
In what follows we will also associate a finite Markov chain with a general QC-type channel 
and investigate the consequences for broadcasting of states and correlations.

\begin{figure}[t]
\begin{center}
\includegraphics[scale=0.22]{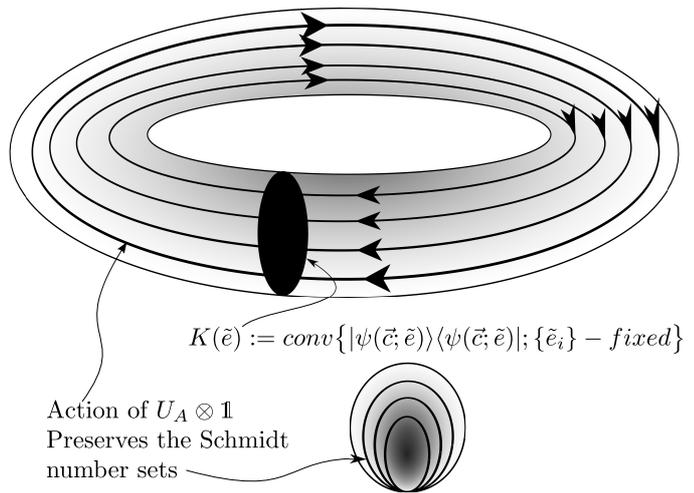}
\caption{\label{fig:pawel}Graphical representation of the set generated by vectors (\ref{Psic}) as a solid torus.
The cross-section represents convex sets $K(\tilde e)$, generated by mixing all the states 
$|\psi(\vec c;\tilde e)\rangle\langle\psi(\vec c;\tilde e)|$ with a fixed Alice's basis $\{\tilde e_i\}$:
$\sum_{\vec c}p(\vec c)|\psi(\vec c;\tilde e)\rangle\langle\psi(\vec c;\tilde e)|$. 
Each $K(\tilde e)$ further contains a hierarchy of convex sets of states with Schmidt number \cite{Maciek}
not greater than $k$, $k=1,\dots,d$.
The action of $U_A\otimes\mathds{1}$ connects different $K(\tilde e)$'s and preserves the Schmidt number sets.}
\end{center}
\end{figure}

The state (\ref{witness_cc}) is obviously a QC state. It will
be a CC state iff there exists a common basis $\{\tilde e_i\}$ such that 
\be\label{cond:cc}
\frac{1}{p_j}\tr_B(\varrho_{AB}\mathds{1}\otimes E_j)=\sum_i p_{i|j}|\tilde e_i\rangle\langle \tilde e_i|,
\ee
for every $j$, where $p_j:=\tr(\varrho_{AB}\mathds{1}\otimes E_j)$ and 
$p_{i|j}:=(1/p_j) \langle \tilde e_i|\tr_B(\varrho_{AB}\mathds{1}\otimes E_j)|\tilde e_i\rangle.$
Condition (\ref{cond:cc}) means that all the Alice residual states, to which Bob steers via his measurement 
\be\label{residual}
\varrho_j^A:=\frac{1}{p_j}\tr_B(\varrho_{AB}\mathds{1}\otimes E_j),
\ee
are simultaneously diagonalizable, or equivalently
\be
[\varrho_j^A,\varrho_{j'}^A]=0 \label{commute}
\ee
for all $j,j'$ (cf. Eq.~(\ref{witness_cc})).

Let us investigate the set $CC(\Lambda)$ of states $\varrho_{AB}$ which solve the above condition, i.e. lead
to a CC state via (\ref{witness_cc}) for a given CC-type channel $\Lambda$. 
We are able to state what follows
\begin{itemize}
\item  Obviously $P_+\in CC(\Lambda)$, by the very assumption (\ref{witness_cc_0}), but  
it also contains mixtures of pure states with the following 
Schmidt decompositions:
\be\label{Psic}
\psi_{AB}(\vec c;\tilde e):=\sum_{i} c_{i}|\tilde e_{i}\rangle_A \otimes |e_{i}^{*}\rangle_B,
\ee
where $\vec c\in\mathbb R_+^d$, $\sum_ic_i^2=1$, $\{\tilde e_i\}$ is some arbitrary basis, and $\{ e^*_i\}$
is the fix basis from Eq.~(\ref{Ej'}). 
Indeed, the states (\ref{residual}) for $|\psi(\vec c;\tilde e)\rangle\langle\psi(\vec c;\tilde e)|$ read:
$p_j\varrho^A_j=\sum_i p^\Lambda_{\phantom{i}j|i}c_i^2 |\tilde e_i\rangle\langle \tilde e_i|$, from which
there appears a stratified structure of convex sets generated by (\ref{Psic}): 
mixing is allowed only within the states with the same, fixed $\{\tilde e_i\}$, thus
generating convex subsets $K(\tilde e)$. Partial unitaries
$U_A\otimes\mathds{1}$ transform between different $K(\tilde e)$'s. 
Furthermore, inside each $K(\tilde e)$ there is a hierarchy of convex sets with increasing 
Schmidt number \cite{Maciek}. This hierarchy is preserved by $U_A\otimes\mathds{1}$.
A schematic representation of this set is given in Fig.~\ref{fig:pawel}.
Note that both $\psi_+$ and its 
local orbit $U_A\otimes U_B\psi_+$ are of the form (\ref{Psic}), as
$U_A\otimes U_B\psi_+=(U_AU_B^T\otimes\mathds{1})\psi_+$ and $U_AU_B^T$ is unitary.
For a general QC-type channel, the states (\ref{Psic}) (for an arbitrary $\{e_i^*\}$) 
will not be in its $CC(\Lambda^{QC})$,
since the residual states $p_j\varrho^A_j=\sum_{i,k} c_ic_k \langle e^*_i|E^{QC}_j e^*_k\rangle|\tilde e_k\rangle\langle \tilde e_i|$
will not in general commute as $E^{QC}_j$'s do not.

\item All CQ (C$_A$Q$_B$) states belong to $CC(\Lambda)$. 
Indeed, substituting into Eq.~(\ref{witness_cc}) an arbitrary C$_A$Q$_B$ state
\be
\varrho_{AB}=\sum_ip_i|\tilde e_i\rangle_A\langle \tilde e_i|\otimes\sigma_i^B, 
\ee
we obtain from Eq.~(\ref{residual}) that
$\varrho_j^A=\sum_i(p_i/p_j)\tr(\sigma_i^B E_j)|\tilde e_i\rangle\langle \tilde e_i|$.
Since $\tr(\sigma_i^BE_j)=p_{j|i}$ is the conditional probability of obtaining result $j$
when measuring POVM $\{E_j\}$ in the state $\sigma_i^B$, from Bayes Theorem
$(p_i/p_j)\tr(\sigma_i^B E_j)=p_{i|j}$ is the needed conditional probability (cf. Eq.~(\ref{cond:cc})).
A schematic representation of the set of CQ states is given in Fig.~\ref{fig:cq}.
For a general QC-type channel, CQ states 
are also in its $CC(\Lambda^{QC})$.
\item Similarly to the set of all CC states, $CC(\Lambda)$ is not convex, which is easily seen from 
the bi-linearity of the condition (\ref{commute}),
but is star-shaped with respect to the maximally mixed state $\mathds{1}/d^2$: if $\varrho_{AB}\in CC(\Lambda)$, then 
\be
\tilde\varrho_{AB}:=\lambda\varrho_{AB} +(1-\lambda)\frac{\mathds{1}_A\otimes\mathds{1}_B}{d^2}\in CC(\Lambda). 
\ee
This follows immediately form (\ref{residual}), as
$\tilde p_j\tilde\varrho_j^A=\lambda p_j\varrho_j^A + (1-\lambda) (\tr E_j)\mathds{1}/d^2$
and $\tilde\varrho_j^A$ pairwise commute iff $\varrho_j^A$ do so. The same is true for $CC(\Lambda^{QC})$
for a general CQ-type channel.
\end{itemize}

\begin{figure}[t]
\begin{center}
\includegraphics[scale=0.22]{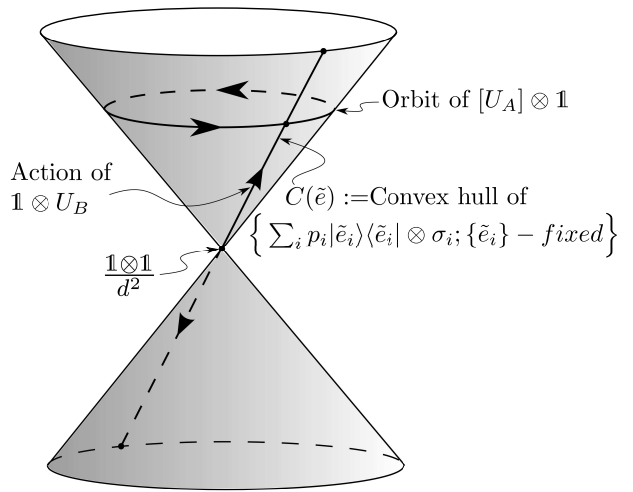}
\caption{\label{fig:cq}Graphical representation of the set of CQ states as a conical surface. 
The generators of the cone represent convex subsets $C(\tilde e)$, obtained 
by mixing all the states of the form $\sum_ip_i|\tilde e_i\rangle\langle \tilde e_i|\otimes\sigma_i$ 
with a fixed Alice's basis $\{\tilde e_i\}$. The local group $\mathds{1}\otimes U_B$ acts along each such a subset.
Different subsets are connected by the action of $[U_A]\otimes \mathds{1}$, where $[U_A]$ denotes the class of $U_A$ modulo a permutation matrix
(evidently the action of Alice's permutations conserve each $C(\tilde e)$). The whole set is star-shaped with respect to 
$(\mathds{1}\otimes\mathds{1})/d^2$.}
\end{center}
\end{figure}

We do not know at this stage if the above conditions fully characterize $CC(\Lambda)$ for a given 
 $\Lambda^{CC}$ and we 
postpone the question of its full characterization for a future research. 
Obviously, one can define the set $CC(\Lambda)$ for any channel $\Lambda$,
however in the light of Thm.~\ref{thm:main} for QC- and CC-type channels  it possesses an interesting interpretation:   
If we think of Alice and Bob as of Environment and System respectively, then 
$CC(\Lambda)$ is the set of those initial System-Environment  states $\varrho_{AB}$ that after the measurement,
described by Thm.~\ref{thm:main} by every $\Lambda^{QC}$, and tracing out the System
lead to Apparatus-Environment states with no quantum correlations, i.e. the Apparatus becomes quantumly 
de-correlated from the Environment.

We now investigate if a QC-type channel $\Lambda^{QC}$
can be used (after a modification)
for state broadcasting \cite{state-broadcasting}. 
We first study a relaxed scenario where we broadcast only eigenvalues, or in other words a classical
probability distribution:
For a given state $\varrho_*$ we are looking for a broadcast state $\sigma_{AB}$ such that 
$U_A\tr_B\sigma_{AB}U_A^\dagger=\varrho_*=U_B\tr_A\sigma_{AB}U_B^\dagger$ for some unitaries $U_A, U_B$.
We will call such a relax broadcasting \emph{spectrum broadcasting} and 
the usual state broadcasting in the sense of Ref.~\cite{state-broadcasting}---\emph{full broadcasting}.
In what follows we prove
\bet\label{thm:broadcast}
For any QC-type  channel $\Lambda^{QC}$ and any orthonormal basis $\{\phi_j\}$ there exists at least one state
$\varrho_*(\phi)$, diagonal in $\{\phi_j\}$, which is $N$-copy spectrum-broadcastable using $\Lambda^{QC}$. 
The state $\varrho_*(e)$, diagonal in the channel's basis
$\{e_j\}$ (cf. Eq.~(\ref{witness})), is also $N$-copy fully broadcastable.
\eet
\begin{proof}
By Thm.~\ref{thm:main} and Eq.~(\ref{MP}) every QC-type channel is a quantum-to-classical 
measurement map. A sufficient condition for spectrum-broadcastability of a state
\be\label{rho*}
\varrho(\phi):=\sum_j\lambda_{j}(\phi)|\phi_j\rangle\langle \phi_j|
\ee
is then that its eigenvalues $\vec\lambda(\phi)$ are preserved by the measurement, i.e.
\be\label{suff}
\tr\big(\varrho(\phi)E_i\big)=\lambda_{i}(\phi)
\ee
for every $i$. This is equivalent to the following eigenvalue problem
\be\label{broadcast}
\sum_jp_{i|j}(\phi)\lambda_{j}(\phi)=\lambda_{i}(\phi)
\ee
for a $d\times d$ stochastic matrix 
\be\label{pphi}
P(\phi):=[p_{i|j}(\phi)],\quad p_{i|j}(\phi):=\langle\phi_j|E_i \phi_j\rangle.
\ee
That this is a stochastic matrix, or equivalently a matrix of conditional probabilities,
follows from the fact that $E_i$'s form a POVM by Eqs.~(\ref{sum_sigma}) and (\ref{Ei}): 
\be
\sum_ip_{i|j}(\phi)=\langle\phi_j|\Big(\sum_iE_i\Big) \phi_j\rangle
=\langle\phi_j|\phi_j\rangle=1
\ee
for every $j$. By the celebrated Perron-Frobenius Theorem \cite{HornJonson}
the above eigenvalue problem (\ref{broadcast}) has at least one non-negative, normalized solution $\vec\lambda_*(\phi)$,
from which we construct through Eq.~(\ref{rho*}) the desired state $\varrho_*(\phi)$. 
Moreover, this solution is unique iff the matrix $P(\phi)=[p_{i|j}(\phi)]$ is \emph{primitive}, 
i.e. is irreducible and possesses exactly one eigenvector of the maximum modulus (equal to $1$ in 
our case), which in turn is equivalent to that 
all the entries of the $(d^2-2d+2)$-th power of $P(\phi)$ are non-zero \cite{HornJonson}. 
We now construct from $\Lambda^{QC}$ a new channel (cf. Eq.~(\ref{MP}))
\be\label{LN}
\Lambda^{(N)}(\varrho):=\sum_i \tr(\varrho E_i)|e_i\rangle\langle e_i|\otimes\cdots\otimes|e_i\rangle\langle e_i|,
\ee
which by condition (\ref{suff}) $N$-copy spectrum-broadcasts the 
state $\varrho_*(\phi)$ (or equivalently $N$-copy broadcast its eigenvalues).

Since the basis $\{\phi_j\}$ above was arbitrary, we obtain from the Perron-Frobenius Theorem that
there exists a spectrum-broadcastable state in any basis (the states in different bases can be equal though, e.g. when
the bases differ only by a permutation). For the basis $\{e_i\}$, associated with
$\Lambda^{QC}$ by the QC-condition (\ref{witness}), the corresponding state $\varrho_*(e)$ will be a fixed
point  of $\Lambda^{QC}$: $\Lambda^{QC}(\varrho_*(e))=\varrho_*(e)$ by Eqs.~(\ref{MP}) and (\ref{suff}). Thus 
$\Lambda^{(N)}(\varrho_*(e))=\sum_j \lambda_{*j}(e)|e_j\rangle\langle e_j|\otimes\cdots\otimes|e_j\rangle\langle e_j|$
is a full $N$-copy broadcast state of $\varrho_*(e)$.
\end{proof}

All the above obviously applies to CC-type channels, as a subclass of QC-type ones.
However, as already mentioned, with any CC-type channel $\Lambda$ there is a naturally associated stochastic matrix $p^\Lambda_{\phantom{i}j|i}$ through Eqs.~(\ref{witness_cc_0},\ref{p'}), without a need of an additional basis ($\{e^*_i\}$
of Eq.~(\ref{Ej'}) plays its role). The corresponding solution $\vec\lambda^\Lambda_{*}\equiv\vec\lambda_*(e^*)$ of (\ref{broadcast}), 
$\sum_ip^\Lambda_{\phantom{i}j|i}\lambda^\Lambda_{*i}=\lambda^\Lambda_{*j}$, and the state $\varrho^\Lambda_*\equiv\varrho_*(e^*)$ are now intrinsic 
characteristics of the channel. Note that 
for an arbitrary basis $\{\phi_j\}$, Eq.~(\ref{broadcast}) reads
\be
\sum_{i,k} p^\Lambda_{\phantom{i}j|i}|U_{ik}|^2\lambda_{k}(\phi)=\lambda_j(\phi),
\ee 
where $\phi_j=:Ue_j^*$ and $|U_{ik}|^2:=|\langle e^*_i|U e^*_k\rangle|^2$ is a doubly-stochastic matrix.
By the Birkhoff Theorem every such a matrix is a convex combination of at most $d^2-2d+2$ distinct permutation
matrices $P_\sigma$, $\sigma\in\mathfrak S_d$ \cite{HornJonson} and hence
\be
p_{i|j}(\phi)=\sum_{\sigma\in\mathfrak S_d}p_\sigma \sum_kp^\Lambda_{\phantom{i}i|k}(P_{\sigma})_{kj}
=\sum_{\sigma\in\mathfrak S_d}p_\sigma p^\Lambda_{\phantom{i}i|\sigma^{-1}(j)},
\ee
while for a general QC-type channel there will also be a ``coherent'' part:
\be
p_{i|j}(U\phi)=\sum_{\sigma\in\mathfrak S_d}p_\sigma p_{\phantom{i}i|\sigma^{-1}(j)}(\phi)
+\sum_{k\ne l}U_{kj}^*U_{lj}\langle k|E_i|l\rangle.
\ee

The existence of fully broadcastable state(s) $\varrho_*(e)$ for any QC-type channel is in some way surprising,
as the measurements described by such channels are in general not von Neumann measurements, but POVMs (cf. Eq.~(\ref{MP})).
The existence of a whole family of spectrum-broadcastable states is perhaps even more surprising. Note, however,
that spectrum-broadcastability is a far weaker condition than full state broadcasting.
By the same reason, although the broadcasting channel $\Lambda^{(N)}$ is the same for every basis---it depends only
on $\Lambda$, we do not contradict the no-go theorem for state broadcasting from Ref.~\cite{state-broadcasting}.

From a probabilistic point of view, the existence of (spectrum-)broadcastable states
follows form the fact that one can associate a finite Markov process with the problem 
through Eq.~(\ref{pphi}), 
and by the Perron-Frobenius Theorem each such a process possesses a 
stationary distribution. The (spectrum-)broadcastable states are constructed precisely from this
distribution.

Let us continue the above analysis and study the implications of the Ergodic Theorem for finite Markov chains \cite{HornJonson}: 
For a stochastic matrix $P$, there exists a limit $P^{\infty}:=\lim_{r\to\infty}P^r$ iff $P$ is primitive. 
The limit is given by 
\be\label{assymptotic}
P^\infty_{ij}=\lambda_{*i}1_j,\
\ee 
where $\lambda_{*i}$ is the stationary distribution (Perron vector) of $P$ (cf. Eq.~(\ref{broadcast})) and $\vec 1:=(1,\dots,1)$.
Note that the limiting matrix elements are the same for each column index $i$:
Asymptotically the probability for the process to be at site $j$ does not depend on the initial site $i$.
As a consequence, the limiting distribution of the process $p^\infty_i:=\sum_jP^\infty_{ij}p_{j}$ will not depend on 
the initial distribution $p_{j}$: 
\be\label{const}
\sum_jP^\infty_{ij}p_{j}=\lambda_{*i}.
\ee
Consider now the $r$-th power of a QC-type channel $\Lambda$: 
\be\label{Lr}
\Lambda^r(\varrho)=\sum_{i,j}P(e)^{r-1}_{ij}\tr(\varrho E_j)|e_i\rangle\langle e_i|, 
\ee
where $P(e)$ is defined through Eq.~(\ref{pphi}). By the Ergodic Theorem, the limit $\lim_{r\to\infty}\Lambda^r=:\Lambda^\infty$ 
exists iff the matrix $P(e)$ is primitive. By Eqs.~(\ref{assymptotic}) and (\ref{Lr}), $\Lambda^\infty$ is then 
a constant channel, analogously to (\ref{const})
\be
\Lambda^\infty(\varrho)=\varrho_*(e)
\ee
for any state $\varrho$. Indeed, $\Lambda^\infty(\varrho)=
\sum_{i,j} P(e)^\infty_{ij}\tr(\varrho E_j)|e_i\rangle\langle e_i|
=(\tr\varrho) \sum_{i} \lambda_{*i}  |e_i\rangle\langle e_i|=\varrho_*(e)$ (cf. Eq.~(\ref{rho*})).
As a consequence, $\Lambda^\infty$ breaks all correlations:
$\mathds{1}\otimes\Lambda^\infty(\varrho_{AB})=\varrho_B\otimes\varrho_*$, $\varrho_B:=\tr_B\varrho_{AB}$.

An interesting situation arises when Eq.~(\ref{broadcast}) has more than one solution,
i.e. when a QC-type channel $\Lambda^{QC}$ (spectrum-)broadcasts \cite{note} more than one state.
Probabilistically, this means that the Markov process, corresponding to $\Lambda^{QC}$ and a context
$\{\phi_i\}$ through Eq.~(\ref{pphi}), possesses more than one stationary distribution.
This happens when the process splits into two or more disconnected processes. 
Algebraically this means that the transition matrix $P(\phi)=[p_{i|j}(\phi)]$ 
is, modulo a column permutation, a direct sum of two or more primitive stochastic matrices
\begin{equation}
P_{d \times d}(\phi) = P_{k \times k}^{(1)}(\phi) \oplus P_{(d-k) \times (d-k)}^{(2)}(\phi).
\label{P-matrix}
\end{equation}
According to the Perron-Frobenius Theorem, each of the blocks has a unique 
Perron vector $\vec{\lambda}_*^{(1)}(\phi)$,
$\vec{\lambda}_*^{(2)}(\phi)$ correspondingly (each of them is normalized). 
Clearly, any $d$-dimensional vector of the form
$\vec{\lambda}_*= p \vec{\lambda}_*^{(1)}\oplus (1-p) \vec{\lambda}_*^{(2)}$ 
is again an eigenvalue-$1$ eigenvector of $P(\phi)$
for any $p \in [0,1]$. We shall denote the corresponding 
states by $\varrho_*^{(1)}(\phi):=diag[\lambda^{(1)}_{1},..,\lambda^{(1)}_{k},0,..,0]$,
$\varrho_*^{(2)}(\phi):=diag[0,..,0,\lambda^{(2)}_{k+1},..,\lambda^{(2)}_{d}]$. 
This is an example of the case where any state from the convex combination 
$p\varrho_*^{(1)} + (1-p)\varrho_*^{(2)}$
can be (spectrum-)broadcasted. Clearly, this example generalizes to more than 
a binary combination of states if the matrix $P(\phi)$ decomposes 
into more than two components: if the number of terms (degeneracy) in Eq.~(\ref{P-matrix})
is $D$, there exists a $D$-dimensional simplex of states (spectrum-)broadcastable by
$\Lambda^{QC}$ (cf. Eq.~(\ref{LN})). The most degenerate case is of course when $D=d$, i.e. when the transition 
matrix $P(\phi)=\mathds{1}$, so that the Markov process is trivial---there are no transitions between the sites,
which happens when the POVM is in fact a von Neumann measurement in $\{\phi_i\}$: $E_i=|\phi_i\rangle\langle \phi_i|$.
  
One can continue the above analysis and consider local broadcasting of correlations.
From the general No-Local-Broadcasting Theorem from Ref.~\cite{broadcasting}, we know that the only locally 
broadcastable states are the CC ones. 
Let us thus consider a family of CC states, build from the stationary solutions $\varrho_*^{(m)}(\phi)$ corresponding 
to a degenerate transition matrix $P(\phi)$:
\beqa\label{CCfamily}
&& \varrho_{*AB}(\pi;\phi):=\sum_{m,n=1}^{D} \pi_{mn} \varrho_*^{(m)}(\phi) \otimes \varrho_*^{(n)}(\phi)\nonumber\\
&& =\sum_{i,j=1}^d\sum_{m,n=1}^D\pi_{mn}\lambda_{*i}^{(m)}\lambda_{*j}^{(n)}|\phi_i\rangle\langle \phi_i|\otimes|\phi_j\rangle\langle \phi_j|.
\eeqa 
Applying to $\varrho_{*AB}(\pi;\phi)$ the product channel $\Lambda^{(N)} \otimes \Lambda^{(N)} $, 
where $\Lambda^{(N)}$ is defined in (\ref{LN}), one achieves 
a local $N$-copy (spectrum-)broadcasting  \cite{note} of the classical correlations:
$[\Lambda^{(N)} \otimes \Lambda^{(N)}]\varrho_{*AB}(\pi;\phi)=\sigma_{A_1\dots A_N B_1\dots B_N}(\pi;\phi)$ and
all the bipartite reductions $\sigma_{A_rB_r}(\pi;\phi)$ are (unitary equivalent/)equal to $\varrho_{*AB}(\pi;\phi)$.
We present a concrete example of this broadcasting scheme in the Appendix, Eqs.~(\ref{P-1},\ref{P-2}),
while a version with two different channels will be studied in what follows.

Let us now assume that two different channels $\Lambda_A, \Lambda_B$ satisfy the assumptions of Thm.~{\ref{thm:main} on
Alice's and Bob's side respectively, i.e.
\beqa
&& \Lambda_A\otimes \mathds{1} (P_+)=\sum_i p_i^A |e_i\rangle_A\langle e_i|\otimes \sigma_i^B,\label{LA}\\
&& \mathds{1}\otimes\Lambda_B (P_+) =\sum_j p_j^B \sigma_j^A\otimes |f_j\rangle_B\langle f_j|.\label{LB}
\eeqa
Then one easily proves
\begin{corollary}
If $\Lambda_A\otimes \mathds{1} (P_+)$ and $\mathds{1}\otimes\Lambda_B(P_+)$ are C$_A$Q$_B$ and Q$_A$C$_B$ states respectively, then 
$\Lambda_A\otimes\Lambda_B(\varrho_{AB})$ is a CC state for any state $\varrho_{AB}$. 
\end{corollary}
\begin{proof}
Indeed, from the proof of Thm.~\ref{thm:main}  it follows that $\Lambda_A$ and $\Lambda_B$ are measurement
maps (cf. Eq.~(\ref{MP})) on Alice and Bob sides respectively, defined by POVM elements 
\be
E_i^A:=dp_i^A(\sigma_i^{B})^T,\quad E_j^B:=dp_j^B(\sigma_j^{A})^T. 
\ee
Thus
\beqa
&& \Lambda_A\otimes\Lambda_B(\varrho_{AB})=(\Lambda_A\otimes\mathds{1})(\mathds{1}\otimes\Lambda_B)\varrho_{AB}\nonumber\\
&& =\sum_{i,j}\tr_A\big[E_i^A\tr_B(\varrho_{AB}\mathds{1}\otimes E_j^B)\big]
|e_i\rangle\langle e_i|\otimes |f_j\rangle\langle f_j|\nonumber\\
&& =\sum_{i,j}\tr(\varrho_{AB} E_i^A\otimes E_j^B)|e_i\rangle\langle e_i|\otimes |f_j\rangle\langle f_j|.
\eeqa
\end{proof}

The analysis of state broadcasting may be repeated in the present scenario as well. 
Since $\mathds{1}_{AB}\otimes\Lambda_{A'}^{CQ}\otimes\Lambda_{B'}^{QC}(P_+^{ABA'B'})=
\big[\mathds{1}_A\otimes\Lambda_{A'}^{CQ}(P_+^{AA'})\big]\otimes\big[\mathds{1}_B\otimes\Lambda_{B'}^{QC}(P_+^{BB'})\big]$,
the channel $\Lambda_{A}^{CQ}\otimes\Lambda_{B}^{QC}$ is of a Q$_{A'B}$C$_{AB'}$-type.
From Thm.~\ref{thm:broadcast} it then immediately follows that for any
basis $\{\phi^{AB}_{\alpha}\}$, $\alpha=1,\dots,d_Ad_B$, of $\mathcal H_A\otimes\mathcal H_B$ 
(the spaces $\mathcal H_A$, $\mathcal H_B$
need not be the same now) there exists
a state $\varrho_{*AB}(\phi^{AB})$, build from a stationary distribution of the stochastic matrix (\ref{pphi}) 
\be\label{PAB}
P^{AB}(\phi^{AB})_{\alpha\beta}:=\langle\phi^{AB}_{\beta}|E^A_i\otimes E^B_j \phi^{AB}_{\beta}\rangle,
\ee
$\alpha:=(ij)$, and locally (spectrum-)broadcastable 
through $\Lambda^{(N)}_A\otimes\Lambda^{(N)}_B$ (cf. Eq.~(\ref{LN})).
Note that the basis $\{\phi^{AB}_{\alpha}\}$ need not be a product one in general.

However, for a product basis $\phi^{AB}_{\alpha}\equiv\phi^{A}_{i}\otimes\phi^{B}_{j}$ one can say more.
The matrix $P^{AB}(\phi^{AB})$ is then a product as well: $P^{AB}(\phi^{AB})=P^{A}(\phi^A)\otimes P^{B}(\phi^B)$ and $P^{AB}(\phi^{AB})$ is primitive iff both $P^{A}(\phi^A)$ and $P^{B}(\phi^B)$ are, i.e. 
$\Lambda_A^{(N)}$, $\Lambda_B^{(N)}$ spectrum-broadcast only one state each. In such a case, the product state $\varrho_{*AB}(\phi^{AB})=\varrho_{*A}(\phi^A)\otimes\varrho_{*B}(\phi^B)$ 
is the only state that can be spectrum-broadcasted and there is no local broadcasting of classical correlations---the 
spectrum of $\varrho_{*AB}(\phi^{AB})$ is a product, $\lambda_{*ij}(\phi^{AB})= \lambda_{*i}(\phi^{A})\lambda_{*j}(\phi^{B})$. 
If, however, at least one channel spectrum-broadcasts 
more than one state, then there exists a family of locally spectrum-broadcastable correlated CC states, built analogously as in  Eq.~(\ref{CCfamily}): 
$\varrho_{*AB}(\pi;\phi^A,\phi^B):=\sum_{m,n=1}^{D_A,D_B} \pi_{mn} \varrho_{*A}^{(m)}(\phi^A) \otimes \varrho_{*B}^{(n)}(\phi^B)$. 
A concrete example of such a situation is presented in the Appendix, Eqs.~(\ref{PA},\ref{PB}).
When it comes to local full state broadcasting, by Thm.~\ref{thm:broadcast} it is guaranteed for
$\varrho_{*AB}(e,f)$, which is a CC state in the bases $\{e_i\}$, $\{f_j\}$ (cf. Eqs.~(\ref{LA},\ref{LB})),
in accordance with the general results of Ref.~\cite{broadcasting}. Again, if both matrices $P^A(e)$ and $P^B(f)$  
are primitive, $\varrho_{*AB}(e,f)$ is a product state with no correlations. However, if at least one 
$P^A(e)$ or $P^B(f)$ is not primitive, by the above construction there will be a family
of locally broadcastable correlated CC states $\varrho_{*AB}(\pi;e,f)$.

Before we conclude, let us digress on a nature of some multipartite QC states. 
We assume that e.g. Bob holds two (possibly different)
subsystems and that the joint state is Q$_{A}$C$_{BB'}$, that is
\be\label{QAACBB}
\varrho_{ABB'}=\sum_\alpha p_\alpha \sigma_\alpha^{A}\otimes|e_\alpha\rangle_{BB'}\langle e_\alpha|,
\ee
where $\{e_\alpha\}$ is a basis in $\mathcal H_B\otimes \mathcal H_{B'}$, labeled by $\alpha$. 
It is not necessarily a product basis--for the definition of a  Q$_{A}$C$_{BB'}$ state
it is enough that it is orthonormal. What is interesting is that simultaneously
forcing both reductions $\varrho_{AB}:=\tr_B\varrho_{ABB'}$ and
$\varrho_{AB'}:=\tr_{B'}\varrho_{ABB'}$ to be Q$_{A}$C$_{B}$ and Q$_{A}$C$_{B'}$ respectively:
\beqa
&& \varrho_{AB}=\sum_i \lambda_i \varrho_i^{A}\otimes |e_i\rangle_B\langle e_i|,\\
&& \varrho_{AB'}=\sum_{i'} \pi_{i'} \tau_{i'}^{A}\otimes |f_{i'}\rangle_{B'}\langle f_{i'}|
\eeqa 
does not force $\varrho_{ABB'}$ to be Q$_{AB}$C$_{B'}$ and Q$_{AB'}$C$_{B}$ simultaneously
(we may label such a class by Q$_{A}$C$_{B}$C$_{B'}$), 
i.e. $\{e_\alpha\}$ in (\ref{QAACBB}) still need not be a product basis. As a simple example  
consider $\mathcal H_B=\mathcal H_{B'}=\mathbb C^2$, 
and $\{e_\alpha\}_{\alpha=1,\dots,4}$--the Bell basis.
Then obviously both reductions $\varrho_{AB}$, $\varrho_{AB'}$ are product, 
$1/2 \big(\sum_\alpha p_\alpha\sigma_\alpha\big)\otimes\mathds{1}$, and hence 
trivially Q$_{A}$C$_{B}$ and Q$_{A}$C$_{B'}$, but the whole state 
$\varrho_{ABB'}$ is not Q$_{A}$C$_B$C$_{B'}$.

In some sense a converse of the above observation is also true: there exist Q$_{A}$C$_{BB'}$ states
with a product basis on $BB'$, which are nevertheless not Q$_{A}$C$_{B}$C$_{B'}$, or, equivalently,
both reductions $\tr_B\varrho_{ABB'}$ and $\tr_{B'}\varrho_{ABB'}$ are not Q$_{A}$C$_{B}$ and 
Q$_{A}$C$_{B'}$ respectively.
As an example of such a state consider $\mathcal H_B=\mathcal H_{B'}=\mathbb C^3$, 
and choose as $\{e_\alpha\}_{\alpha=1,\dots,9}$ in (\ref{QAACBB}) the ``nonlocality without entanglement'' 
$3\otimes 3$ basis from Ref.~\cite{nonlocality}. Then both $\tr_B\varrho_{ABB'}$ and $\tr_{B'}\varrho_{ABB'}$
will contain an overcomplete set on $B$ and $B'$ side respectively.

In conclusion, we have provided a refinement of the characterization of entanglement breaking channels
from Ref.~\cite{Michal} to more general quantum correlations 
and connected it to measurement maps, quantum state/correlations broadcasting, and finite Markov chains. 
We have considered two classes of channels---the ones that
(i) break quantum correlations by turning them into the QC form and (ii) that 
fully break quantum correlations by turning them into CC ones.
We have shown that a channel belongs to the first class iff it turns a
maximally entangled state into a QC state or equivalently it is represented by a 
measure-and-prepare scheme,
where the outcomes of a POVM measurement are followed by a preparation of states from 
some specific orthonormal basis. In other words, it is a quantum-to-classical measurement map
(i.e. it gives the state of the Apparatus after tracing the System).

Surprisingly, a similar question in the case of the second class of channels becomes 
even more interesting: the analogy to entanglement-breaking channels now fails and one 
cannot characterize the channels from the second class only by their actions on the maximally entangled state.
However, a characterization from a different perspective seems possible. 
First of all, it turns out that the POVMs, constituting the 
channels, are mutually commuting and arise from a stochastic matrix,
thus making a connection to finite Markov chains. 
Second, the set of bipartite states that are mapped into the CC 
form is more complicated. 

Our analysis of the ability to broadcast quantum 
states and correlations by QC-type channels reveals 
an interesting application of the Perron-Frobenius Theorem. 
The existence of a family of spectrum-broadcastable states and at least one fully broadcastable state, 
even if the POVM measurement is not 
of the von Neumann type, follows from the fact that each finite Markov process
possesses a stationary distribution. This broadcasting scheme, albeit in general substantially weaker than the standard broadcasting
of e.g. Refs.~\cite{broadcasting,state-broadcasting}, surprisingly goes beyond the simple C-NOT scenario.
The connection between broadcasting and finite Markov chains is, to our knowledge, quite unexpected
and will be a subject of a further research.

In fact, perfect broadcasting operations applied so far
corresponded to a scenario where to a given input CC state $\varrho_{AB}=\sum_{i,j} p_{ij} |i\rangle\langle 
i| \otimes |j\rangle\langle j|$ one locally 
applies the generalized C-NOT gates $ U |i\rangle |j\rangle := 
|i\rangle |i \oplus j \rangle $.
Application of the Perron-Frobenius Theorem presented in this work goes beyond this 
simple scenario. 

We believe that the current work opens new perspectives
for an analysis of the measurement problem and state/correlations broadcasting.
Especially interesting seems possibility to study quantum 
decoherence in terms of broadcasting.

We thank  D. Chru\'sci\'nski, A. Grudka, M. Horodecki, M. Ku\'s, A. Rutkowski,  K. \.Zyczkowski
for discussions and remarks,  F. G. S. L. Brand\~ao for a question on the CQ limitation of Thm.~\ref{thm:main},
and J. Wehr for providing useful information on Markov chains.
This work is supported by the ERC Advanced Grant QOLAPS and  
National Science Centre project Maestro DEC-2011/02/A/ST2/00305.
JKK acknowledges the financial support of the QCS and TOQUATA Projects and would like to thank National Quantum Information Centre (KCIK) 
in Gda\'nsk and Center for Theoretical Physics of The Polish Academy of Sciences for hospitality.

\appendix*
\section{}\label{app}
Consider the following example.  Let
\begin{equation}
 P^{(1)}:=
\left[ \begin{array}{clrrr}
 0 & \frac{1}{2} &\frac{1}{2} \\
\frac{1}{2}  & \frac{1}{2} & \frac{1}{2} \\
\frac{1}{2} &  0 & 0   \\
\end{array}
 \right ]
 \label{P-1}
 \end{equation}
for some fixed basis $\{\phi_i\}$ and $P^{(2)}$ be an arbitrary irreducible bistochastic matrix
on $\mathbb R^{3}$, say:
\begin{equation}
P^{(2)}:=
\left[ \begin{array}{clrrr}
 \frac{1}{8} & \frac{3}{8} &\frac{1}{2} \\
\frac{3}{8}  & 0 & \frac{5}{8} \\
\frac{1}{2} &  \frac{5}{8} & 0   \\
\end{array}
 \right ]
 \label{P-2}
\end{equation}
for the same basis. Since we know that any matrix $A \in {\cal M}_{d \times d}(\mathbb R)$ with
non-negative elements is irreducible iff $(\mathds{1} + A)^{d-1}$ has all elements non-negative,
we may easily check that both matrices are irreducible.
The unique Perron vector of $P^{(1)}$ is just
$\vec{\lambda}^{(1)}=[\frac{1}{3},\frac{1}{6},\frac{1}{2}]^{T}$.
The unique eigenvector of the irreducible bistochastic matrix is of course
$\vec{\lambda}^{(2)}=[\frac{1}{3},\frac{1}{3},\frac{1}{3}]^{T}$.
Consider now the stochastic matrix $P:=P^{(1)} \oplus P^{(2)}$  on
$\mathbb R^{6}$
Then any state of the form $\varrho_{*AB}(\pi)=\sum_{m,n=1}^{2} \pi_{mn} \varrho_{*}^{(m)}
\otimes \varrho_{*}^{(n)}$ with $\varrho_{*}^{(1)}:=diag[\frac{1}{3},\frac{1}{6},\frac{1}{2}]$ and
$\varrho_{*}^{(2)}:=diag[\frac{1}{3},\frac{1}{3},\frac{1}{3}]$
can be spectrum/full broadcasted by the product of the channels $\Lambda^{(N)}$, 
defined in Eq.~(\ref{LN}).

Even simpler example with two different channels can be constructed to illustrate spectrum/full
broadcasting of correlations. Namely, consider two bistochastic matrices
of the form:
\begin{equation}
P^A:=\left[ \begin{array}{clrrr}
 0 & \frac{1}{2} &\frac{1}{2} \\
 0 & \frac{1}{2} & \frac{1}{2} \\
1 &  0 & 0   \\
\end{array}
 \right ]
 \label{PA}
 \end{equation}
and
\begin{equation}
P^B :=
\left[ \begin{array}{clrrr}
\frac{2}{3} & 0 &\frac{1}{3} \\
\frac{1}{3}& 0 & \frac{2}{3} \\
0 &  1 & 0   \\
\end{array}
\right ]
\label{PB}
\end{equation}
for some basis $\{\phi_i\}$. They are clearly reducible.
Finding their Perron vectors and
defining $\varrho_{*AB}(\pi):=\sum_{m,n=1}^{2} \pi_{mn} \varrho_{*}^{(m)} \otimes
\varrho_{*B}^{(j)}$
as
$\varrho_{*A}^{(1)}:=diag[0,\frac{1}{2},\frac{1}{2}]$,
$\varrho_{*A}^{(2)}:=[1,0,0]$ and
$\varrho_{*B}^{(1)}:=diag[\frac{1}{2},0,\frac{1}{2}]$,
$\varrho_{*B}^{(2)}:=[0,1,0]$, we see that $\varrho_{*AB}(\pi)$
is locally broadcastable by the map $\Lambda_A^{(N)} \otimes \Lambda_B^{(N)}$
where $\Lambda_A$, $\Lambda_B$ are defined
again through (\ref{LN}).


\begin{thebibliography}{99}
\bibitem{note} Obviously, by Thm.~\ref{thm:broadcast}
whether we consider spectrum broadcasting or full broadcasting depends on 
the considered basis: spectrum-broadcasting is grated for any basis, while
full broadcasting is warranted for (at least) the basis $\{e_i\}$ associated with  the channel by the condition (\ref{witness}).
\end{thebibliography}
\end{document}